\documentclass[pra,twocolumn,amsmath, superscriptaddress,notitlepage,showpacs]{revtex4-1}
\usepackage{amsmath,amsfonts,amssymb,amstext,amscd,amsthm,dsfont,bbm,hyperref}
\usepackage{physics}
\usepackage{color}
\usepackage{soul,xcolor}
\usepackage{gensymb}
\hypersetup{
    colorlinks = true,
    citecolor=blue,
    linkcolor = red,
    anchorcolor = red,
    citecolor = blue,
    filecolor = red,
    pagecolor = red,
    urlcolor = blue,
    }
\usepackage[normalem]{ulem}
\usepackage{graphicx}
\usepackage{bbold}
\usepackage{braket}
\usepackage{placeins}

\newtheorem{Theorem}{Theorem}

\newtheorem{Remark}{Remark}
\newtheorem{Example}{Example}

\begin{document}

\title{Noninvertibility as a requirement for creating a semigroup under convex combinations of channels}
	\author{Vinayak Jagadish}
	\email{vinayak.jagadish@uj.edu.pl}
	\affiliation{Faculty of Physics, Astronomy and Applied Computer Science, Jagiellonian University, 30-348 Krak{\'o}w, Poland}
			\author{R. Srikanth}
	\affiliation{Poornaprajna Institute of Scientific Research,
		Bangalore- 560 080, India}
	\author{Francesco Petruccione}
	\affiliation{Quantum Research Group, School  of Chemistry and Physics,
		University of KwaZulu-Natal, Durban 4001, South Africa}\affiliation{ National Institute for Theoretical and Computational Sciences (NITheCS),  Durban 4001, South Africa}
%\date{} 

\begin{abstract} 
We study the conditions under which a semigroup is obtained upon convex combinations of channels. In particular, we study the set of Pauli and generalized Pauli channels. We find that mixing only semigroups can never produce a semigroup.  Counter-intuitively, we find that for a convex combination to yield a semigroup, most of the input channels have to be noninvertible.
\end{abstract}
%\pacs{03.65.Yz,03.67.-a}
\maketitle  
\section{Introduction} 
No quantum system can truly be considered to be completely isolated from its surroundings~\cite{petruccione}. Ultimately all quantum systems will be subjected to coupling to what is known as their external environment. The theory of open quantum systems has gained renewed attention due to its fundamental aspects as well as the application to new technologies~\cite{haroche_exploring_2006}. Quantum non-Markovianity~\cite{li_concepts_2017,de_vega_dynamics_2017} is a fascinating theme in the study of open quantum dynamics that continues to pose conceptual challenges and surprises.

The finite time evolution $\rho(0) \to \rho(t) = \mathcal{E}[\rho]$ of an open quantum system in a state $\rho(0)$ is characterized by a
completely positive, trace-preserving (CPTP) dynamical map~\cite{sudarshan_stochastic_1961}, also known as the quantum channel~\cite{Quanta77}. One could also have a continuous-in- time description of the system of interest which is given by the differential form of the map~\cite{gorini_completely_1976}, $\dot{\mathcal{E}} = L\mathcal{E}$, with \begin{equation}
L[\rho]=-\frac {\imath}{\hbar} [H,\rho]+\sum_i \gamma_i (t)
\left(G_i\rho G_i^\dagger-\frac {1}{2}\{G_i^\dagger G_i,\rho\}\right),
\label{gksl}
\end{equation}
where $H$ is the effective Hamiltonian, $G_i$ are the noise operators, and $\gamma_i$ are the decoherence rates. Going a step further, 
Eq. (\ref{gksl}) could be generalized to the time-dependent case with the time-local generator $L(t)$. In general, the decoherence rates $\gamma_i(t)$ need not be positive. Quite generally (barring the existence of noninvertibility), maps and generators for the channel are equivalent and interconvertible~\cite{hall2010}. The divisibility property of a dynamical map is as follows:
\begin{equation}
\mathcal{E}(t_b, t_a) = \mathcal{E}(t_b, t)\mathcal{E}(t, t_a),\quad \forall t_b\geq t \geq t_a.
\label{cpdivdef}
\end{equation}
A quantum evolution of a system starting at time $t_a$ is CP divisible if $\forall t,t_b, \mathcal{E}(t_b, t)$ is CP and satisfies the above divisibility condition. Otherwise, the map is CP indivisible. CP indivisibility is one accepted definition for quantum non-Markovian evolution~\cite{rivas_entanglement_2010}. In our present work, quantum Markovianity is identified  with CP divisibility. An equivalent criterion to the indivisibility of the intermediate map is the temporary negativity of the decay rates in the time-local master equation~\cite{hall2010}, which we would make use of in our future discussions.

There has been increasing attention given to the convex combinations of quantum channels. Recently, we considered~\cite{jagadish_convex_2020} convex combinations of the set of Pauli channels, each assumed to be a quantum dynamical semigroup (QDS), and obtained a quantitative measure of the resulting set of Markovian and non-Markovian (CP-indivisible) channels. The semigroup property of maps implies that 
\begin{equation}
\mathcal{E}(t_a, 0)\mathcal{E}( t_b,0 ) = \mathcal{E}( t_a + t_b,0 ).
\label{semigpdef}
\end{equation}
This leads to a time-independent generator $L$ such that
\begin{equation}
\mathcal{E}(t, t_0)=e^{L( t-t_0 )}.
\label{semigpdef0}
\end{equation}
A QDS implies that all the decay rates $\gamma_i$ in Eq. (\ref{gksl}) are constants, independent of time. Traditionally quantum Markovianity (or the absence of memory effects) was attributed to maps obeying the semigroup property, as in Eq. (\ref{semigpdef}), because it arguably represents a quantum extension of the Chapman-Kolmogorov equation in the context of classical Markovianity. However, one can have maps which are CP divisible [Eq. (\ref{cpdivdef})], but lack the semigroup property [Eq. (\ref{semigpdef})]. They are Markovian (as mentioned above), with the time-dependent decay rates $\gamma_i(t)$ being positive at all times. Therefore, one expects that if one were to consider convex combinations of channels which are Markovian but not semigroups, then this would correspondingly result in a larger measure for the set of non-Markovian channels obtained as the output of mixing. However, this turned out not to be the case, as we showed in~\cite{jagadish_nonqds_2020}, by considering the mixing of Pauli channels that are non semigroups. In fact, the fraction of non-Markovian channels is largest when obtained by mixing Pauli semigroups.  These results showed that the set
of CP divisible Pauli channels is not convex, which was quantified with a non-zero measure of the resulting CP-indivisible set.
Extension to qudit channels and the nonconvexity of the set of generalized Pauli Markovian channels was reported in~\cite{siudzinskajpa2020}.

In Eq. (\ref{cpdivdef}), if for a particular instant of time $t = t_{*}$ such that the composition law fails, the map $\mathcal{E}(t_{*},t_i)$ becomes
noninvertible where at the point $t_{*}$, the map $\mathcal{E}(t_b, t_{*}) =  \mathcal{E}(t_b, t_a) \mathcal{E}(t_{*} t_a )^{-1}$ is
undefined. Such points $t_{*}$ are known as singular points~\cite{hou_singularity_2012} of the corresponding channel. Correspondingly, the decay rates go to infinity at those points. It should be noted that despite the presence of singularities in a map, the dynamics can be perfectly regular~\cite{jagadish_measure2_2019, jagadish_initial_2021}.  Convex combinations of channels with singularities have been studied~\cite{utagi_singularities_2021,siudzinska_markovian_2021} recently. For Pauli channels, it was shown that singularities cannot be produced by mixing nonsingular channels. The effects of mixing in creating new singular points, shifting the singularities in time and also the total elimination of a singularity were also reported.

In this paper, we ask the following question. What convex combinations of generalized Pauli channels can lead to a semigroup? We outline the conditions under which a semigroup can be obtained. We find that the mixing coefficients and the associated probability distributions are interconnected in deciding this. Surprisingly, we find that noninvertibility of the input channels is necessary to produce a semigroup.  Section \ref{paulisec} discusses the case of Pauli channels. The results can be easily generalized to the case of generalized Pauli channels which is discussed in Sec. \ref{genpaulisec}. We conclude in Sec. \ref{conc}. 
\section{Pauli Dephasing channels
\label{paulisec}}
Consider the three Pauli channels 
\begin{equation}
\label{paulichanndef}
 \mathcal{E}_i^{p_{i}}(\rho)=(1-p_{i} (t))\rho + p_{i} (t) \sigma_i\rho \sigma_i,
 \end{equation} with varying $p_i (t) ; i= 1, 2,3$.
In what follows, we show that the only two possibilities of creating a semigroup under convex combinations of Pauli channels are\\
 (1) by mixing the same Pauli channel with varying $p_i (t)$, and \\
 (2) a combination of the three different Pauli channels with suitable choices of $p_i (t)$. 
 Here, we see what choices of probability distributions lead to the semigroup.
 \subsection{Mixing same Pauli channels with different $p(t)$ \label{sec:same}}
  Without loss of generality, consider the convex combination $a \mathcal{E}_1^{p}+ (1-a) \mathcal{E}_1^{q}$. 
  The associated differential form of the map is 
  \begin{equation}
\label{}
L(\rho) = \gamma_{1}(\sigma_1\rho\sigma_1-\rho),
\end{equation}
  with the decay rate being
    \begin{equation}
  \label{decaysamemix}
\gamma_{1} = \frac{  (1-a) \dot{p}(t)+a \dot{q}(t)}{1-2 (1-a) p(t)-2 a q(t)}
\end{equation}
For the convex combination to yield a semigroup, $\gamma_1$ should be a constant. To this end, let 
  \begin{equation}
p(t) = f(t) - \frac{q(t)}{\alpha}, \alpha \in \mathbb R.
\end{equation}
 One can now see that 
  \begin{equation}
 f(t) = \frac{1-e^{-ct}}{2} \frac{1}{1-a}, \alpha = \frac{1-a}{a}
\end{equation}
yields a semigroup. We can see that the choice of mixing coefficients and $p(t)$ are interconnected to yield a semigroup. Another important point to notice is that the function $q(t)$ can be arbitrary. This means that one could get semigroups even if one of the channels is non-invertible. 
Let us demonstrate this with an example. Let us consider the convex combination of two channels as
\begin{equation}
\mathcal{E}_\ast(p) = \alpha\mathcal{E}_x^{p} + (1-\alpha) \mathcal{E}_x^{q}.
\label{twomixeq}
\end{equation}
 For equal combinations of the two $\mathcal{E}_x$ channels. i.e., $a=0.5$,
  \begin{equation}
p(t) = 1-e^{-ct} - q(t),
\end{equation}
with $q(t)$ being arbitrary, leads to a semigroup. The differential form of the channel can be evaluated to be 
\begin{equation}
\label{megen}
L(\rho) = \frac{c}{2}(\sigma_x\rho\sigma_x-\rho),
\end{equation}
Note that we work in the Pauli basis $\{\mathbbm{1}, \sigma_i\}$.
The decay rate $c/2$ is a constant which is indicative of a semigroup. 
$q(t)$ can be any function, even with singularities, but still the convex sum would lead to a semigroup. This also can be viewed as an exercise of mixing non-invertible channels. Note that this construction is not restricted to mixing two channels of the same type. One could mix any number of channels, choosing the appropriate probability distributions.

In the remainder of this article, by a (generalized) Pauli channel in the input, we shall refer to a single $d+1$ ``dephasing'' channel, and not a mixture of two or more of the same type of Pauli channel, as was considered in this section.
 
 \subsection{Mixing three different Pauli channels}
 Consider the three Pauli channels given by Eq. (\ref{paulichanndef}), each mixed in proportions of $x_i, i = 1,2,3$ such that $\sum_{i}x_i=1$.  The corresponding time-local master equation for the channel $\sum_{i} x_{i} \mathcal{E}_i^{p_{i}}$ can be evaluated to be
\begin{equation}
\label{eqngen}
L(\rho) =\sum_{i=1}^{3} \gamma_{i}(\sigma_i\rho\sigma_i-\rho),
\end{equation}
and with
\begin{eqnarray}
\label{decayrateseqn}
\gamma_{1} &=& A-B+C, \nonumber\\
\gamma_{2} &=& -A+B+C,\nonumber\\
\gamma_{3}&=& A+B-C,
\label{eq:g1g2g3}
\end{eqnarray}
with
\begin{eqnarray}
A&=&\frac{x_1 \dot{p_1}(t)+x_3 \dot{p_3}(t)}{2[1-2 x_1 p_1(t)-2 x_3 p_3(t)]},\nonumber\\
B &=& \frac{x_2 \dot{p_2}(t)+x_3 \dot{p_3}(t)}{2[1-2 x_2 p_2(t)-2 x_3 p_3(t)]},\nonumber\\
C&=&\frac{x_1 \dot{p_1}(t)+x_2 \dot{p_2}(t)}{2[1-2 x_1 p_1(t)-2 x_2 p_2(t)]}.
\label{eq:ABC}
\end{eqnarray}

For the channel to be a semigroup, all the decay rates $\gamma_i$ need to be constants.  
This entails that $A, B,$ and $C$ must be constants, since $C = \frac{\gamma_{1} + \gamma_{2}}{2}$, etc. It can be shown that the only functional form of $p_j$ that ensures that
$A, B, C$ are constant is that 
\begin{equation}
p_i(t) = \frac{1-e^{-ct}}{4x_i}. 
\label{eq:form}
\end{equation}
The corresponding decay rate for the input channels would be  $\displaystyle \frac{c}{2+e^{c t}(4 x_{i}-2)}$ which diverges for $t^\ast \equiv \frac{1}{c}\ln [\frac{1}{1-2x_i}]$.
Here the semigroup corresponds to $x_i {:=} \frac{1}{2}$. 

\begin{Example}
Consider an equal mixing of the three Pauli channels. All $x_i$ are 1/3 each. Then, all the $p_i (t)$ are the same with $\displaystyle p(t) = \frac{3}{4} (1-e^{-ct})$, which leads to a semigroup. For the Pauli channel $\mathcal{E}_x^{p}$, say the associated decay rate would be $\displaystyle \frac{3 c}{6-2 e^{c t}}$, which diverges for $t^\ast \equiv \frac{\log3}{c}$, indicating that the individual channels are noninvertible. 
\end{Example} 
We remark that the above result generalizes for qubits a result reported in \cite[Remark 2]{siudzinska_markovian_2021}, which shows that a semigroup results via an equal mixing of three noninvertible Pauli channels. The following result gives the general case here for qubits.
\begin{Theorem} 
In order to produce a semigroup by mixing Pauli channels, all three must be mixed. Of these, at least two should be noninvertible.
\label{rem:only1}
\end{Theorem}
\begin{proof}
Suppose that only two Pauli channels are mixed. To that end, without loss of generality, let $x_3 \equiv 0$. As before, for a semigroup output, we require that $A, B$ and $C$ must be constants. With Eq. (\ref{eq:ABC}), we find that the form given by Eq. (\ref{eq:form}) should be used for $p_1(t)$ and $p_2(t)$ to ensure that $C$ is constant. But this would mean that $A$ and $B$ are not constants. Therefore, all three Pauli channels must be mixed in order to produce a semigroup. 

As seen above, the mixing coefficients and the probability functions $p_i (t)$ are interconnected. For an input Pauli channel to be a semigroup, $p (t) = (1-e^{-ct})/2$, and the corresponding mixing coefficient will be 0.5.  The remaining two channels are to be mixed in arbitrary proportions such that their sum is 0.5, which implies per Eq. (\ref{eq:form}) that these two cannot be semigroups. If the input Pauli channel is invertible, but not a semigroup, then $x_i>0.5$. Thus, for the convex combination of three Pauli channels to yield a semigroup, at most one of the individual channels can be invertible (in particular, a semigroup).
\end{proof}
A consequence of this result is that if the three input channels are semigroups, then we should have $x_i = \frac{1}{2}$ for all three Pauli channels, which is obviously ruled out.
\begin{Remark}
	It was shown in~\cite{jagadish_convex_2020} that mixing the three different Pauli semigroups could lead to Markovian output. However, the above result entails that their convex hull cannot contain any semigroup. For instance, mixing the above three Pauli semigroups characterized by the decoherence parameter $c>0$ yields a channel with the  decay rates in Eq. (\ref{eqngen}) all being the same, namely, $\displaystyle \frac{c}{2(2+e^{c t})}$, which is not a constant.
	\end{Remark}
As a matter of fact, it may be recollected that any finite degree of mixing of two Pauli semigroups results in a non-Markovian output \cite{jagadish_convex_2020}. We remark that the above result is not in contradiction with those reported in \cite[Section III \& Section IV]{wudarski2016}, where the mixing of two CPTP projector channels is shown to produce a semigroup. The individual Pauli dephasing channels $\mathcal{E}$ considered here clearly lack this projection property, i.e., $\mathcal{E}^2 \ne \mathcal{E}$. In the same vein, one can show that mixing two Pauli depolarizing channels can be made to yield a semigroup.
\section{Generalized Pauli Channels}
\label{genpaulisec}
Nathanson and Ruskai  introduced~\cite{ruskai2007} a generalization of the Pauli channels to the case of qudit evolution. The construction is based on the concept of {\it mutually unbiased bases} (MUBs). If one has $d+1$ orthonormal bases $\{|\phi_i^{(\alpha)}\rangle,i=0,\ldots,d-1\}$ in $\mathbb{C}^d$ such that  $\langle\phi_i^{(\alpha)}|\phi_j^{(\beta)}\rangle|^2=1/d$ whenever $\alpha\neq\beta$, the bases are MUB. Using the rank-1 projectors $P_i^{(\alpha)}=|\phi_i^{(\alpha)}\rangle\langle\phi_j^{(\alpha)}|$ , one can introduce $d+1$ unitary operators,
\begin{equation}
U_\alpha=\sum_{i=0}^{d-1}\omega^iP_i^{(\alpha)},\qquad \omega=e^{2\pi \imath/d}.
\end{equation}
The time-dependent generalized Pauli channels are defined by~\cite{ruskai2007}
\begin{equation}
\mathcal {E}_{G} = p_0 (t)\mathbbm{1}+\frac{1}{d-1}\sum_{i=1}^{d+1}p_\alpha (t) \mathcal{U}_\alpha,
\label{genpaulidef}
\end{equation}
where $p_\alpha (t)$ is a probability distribution and the $d+1$ CP maps $\mathcal{U}_\alpha$ are 
\begin{equation}
\mathcal{U}_\alpha [\rho]=\sum_{k=1}^{d-1}U_{\alpha}^{k}\rho U_{\alpha}^ {k \thinspace\dagger}.
\end{equation}
One can see that for $d=2$, Eq. (\ref{genpaulidef}) leads to the Pauli channel.
\subsection{Convex Combinations of Generalized Pauli Channels}
For $d$ dimensions, one has $d+1$ generalized Pauli channels. We consider the convex combinations of them and analyze the conditions under which they lead to a semigroup.  Consider the convex combination of $d+1$ generalized Pauli channels ($i =1,...d+1) $,
\begin{equation}
\mathcal {E}_{G \thinspace i} ^{p_{i}}(\rho)=(1-p_{i}(t))\rho + p_{i}(t)\sum_{k=1}^{d-1}U_{i}^{k}\rho U_{i}^ {k \thinspace\dagger},
\end{equation} with varying $p_i (t)$, each mixed in proportions of $x_i$ such that $\sum_{i=1}^{d+1} x_i=1$.  The map $\mathcal {E}_{G \thinspace i}$ satisfies the eigenvalue relation
\begin{equation}
\mathcal {E}_{G \thinspace i}[U_{i}^{k}] = \lambda_i U_{i}^{k} , k = 1, ..., d-1.
\end{equation}
The eigenvalues of the resultant map under convex combination can be evaluated to be
\begin{equation}
\label{eiggenpauli}
\lambda_i = 1-\frac{d}{d-1} \sum_{j\neq i} x_j p_{j}(t).
\end{equation}
For the mixture to yield a semigroup, the eigenvalues $\lambda_i = e^{-rt}$. 
This clearly tells us that the only two possibilities of getting a semigroup are either\\
\begin{enumerate}
\item by mixing the same channel with varyious $p(t)$'s (dealt with in Sec. \ref{sec:same}) or
 \item by mixing all the  $d+1$ channels with suitable corresponding decoherence functions $p(t)$. [Each channel could itself be a mixture of versions of the channel with various $p(t)$.]
\end{enumerate}
The following gives an example of the first point above, generalizing the example for qubits discussed in Sec. \ref{paulisec}.
 \begin{Example}
Consider the mixing of the same channel with two different $p(t)$ and $q(t)$. Without loss of generality, consider the convex combination $a \mathcal {E}_{G \thinspace 1} ^{p} + (1-a) \mathcal {E}_{G \thinspace 1} ^{q} $. The eigenvalues of the map, following Eq. (\ref{eiggenpauli}), are 
\begin{equation}
\lambda_i =  1- \frac{d}{d-1}  p(t) - \frac{d}{d-1} a q(t) + \frac{d}{d-1} a p(t) 
\end{equation}
For a convex combination to yield a semigroup (i.e., $\lambda_i = e^{-ct}$),  one can see that the choice should be 
\begin{equation}
p(t) =  \frac{1-e^{-ct}}{d} \frac{d-1}{1-a} - \frac{ a q(t)}{1-a },
\end{equation}
with $q(t)$ being arbitrary.
This can be arbitrarily extended to the mixing of more than two channels. 
\end{Example}
Consider the case of the mixing of all $d+1$ channels. The eigenvalues $\lambda_i $ of the map should be $ e^{-rt}$ for the result to be a semigroup. From Eq. (\ref{eiggenpauli}), we see that the sum has $d$ terms and, therefore, the only choice for the input channel decoherence functions that can guarantee a semigroup is found to be
\begin{equation}
\label{decfuncgen}
p_{i}(t) =  \frac{d -1}{x_i d^2} (1-e^{-ct}).
\end{equation}
In particular, we remark that the decoherence functions are dependent on the mixing coefficients $x_i$. The corresponding eigenvalue diverges for $t^\ast \equiv \frac{1}{c}\ln [\frac{1}{1-d x_i}]$.
This gives us the following result,  which generalizes Theorem \ref{rem:only1}.

\begin{Theorem}
\label{Remark3}
All $d+1$ qudit Pauli channels must be mixed in order to produce a semigroup. Out of the $d+1$ channels, at least $d$ channels must be noninvertible.
\end{Theorem}
\begin{proof}
In Eq. (\ref{decfuncgen}), the left hand side is at most, 1. It therefore follows that for any channel $j$,
\begin{equation}
x_j \ge \frac{d-1}{d^2},
\label{eq:minx}
\end{equation}
which implies that none of the $d+1$ channels can be mixed with vanishing probability.

For a generalized Pauli channel to be a semigroup, decoherence function $p(t) = \frac{d -1}{d} (1-e^{-ct})$. In light of Eq. (\ref{decfuncgen}), this implies that for the input channels, if $x_i \ge 1/d$, then it is invertible, and in particular a semigroup when equality holds. If $x_i < 1/d$, then the input channel is noninvertible. Suppose one of the input channels (say the $j$th) is invertible, meaning that $x_j \ge \frac{1}{d}$. Then the average probability $x_i$ per channel among the remaining channels is, at most, $\frac{d-1}{d^2}$, which saturates Eq. (\ref{eq:minx}). Thus, if even one more invertible channel is included, then Eq. (\ref{eq:minx}) would be violated.
\end{proof}
An immediate consequence of the above is that mixing all generalized Pauli dephasing semigroups would never result in a semigroup. 
\section{Conclusions}
\label{conc}
In this paper, we studied the set of Pauli and generalized Pauli channels and have shown the most general conditions under which a semigroup is obtained under convex combinations of them. If all the input channels are semigroups, we would never get a semigroup upon convex combination. Counterintuitively, noninvertibility of the individual channels is mandatory to produce a semigroup. A semigroup is considered to be the ``most" Markovian according to the various definitions of quantum Markovianity. Hence, it is quite surprising that for a convex combination to yield a semigroup, most of the input channels have to be even noninvertible.

\section{Acknowledgements:}
V.J. acknowledges financial support by the Foundation for Polish Science
through TEAM-NET project (contract no. POIR.04.04.00-00-17C1/18-00). R.S.   acknowledges the support of Department of Science and Technology (DST), India, Grant No.: MTR/2019/001516.
The work of F.P. is based upon research supported by the South African Research Chair Initiative of the Department of Science and Innovation and National Research Foundation (NRF) (Grant UID: 64812).  


\vspace{-3 mm}


\begin{thebibliography}{10}

\bibitem{petruccione}
Heinz-Peter Breuer and Francesco Petruccione.
\newblock {\em The Theory of Open Quantum Systems}.
\newblock (Oxford University Press, Oxford, 2002).

\bibitem{haroche_exploring_2006}
Serge Haroche and Jean~Michel Raimond.
\newblock {\em {Exploring the Quantum: Atoms, Cavities, and Photons}}.
\newblock (Oxford University Press, Oxford, 2006).

\bibitem{li_concepts_2017}
Li~Li, Michael J.~W. Hall, and Howard~M. Wiseman.
\newblock Concepts of quantum non-{Markovianity}: {A} hierarchy.
\newblock {\em Phys. Rep.}, 759:1--51 (2018).

\bibitem{de_vega_dynamics_2017}
In{\`e}s de~Vega and Daniel Alonso.
\newblock Dynamics of non-{Markovian} open quantum systems.
\newblock {\em Rev. Mod. Phys.}, 89(1):015001 (2017).

\bibitem{sudarshan_stochastic_1961}
E.~C.~G. Sudarshan, P.~M. Mathews, and Jayaseetha Rau.
\newblock Stochastic {Dynamics} of {Quantum}-{Mechanical} {Systems}.
\newblock {\em Phys. Rev.}, 121(3):920--924 (1961).

\bibitem{Quanta77}
Vinayak Jagadish and Francesco Petruccione.
\newblock An invitation to quantum channels.
\newblock {\em Quanta}, 7(1):54--67 (2018).

\bibitem{gorini_completely_1976}
Vittorio Gorini, Andrzej Kossakowski, and E.~C.~G. Sudarshan.
\newblock Completely positive dynamical semigroups of {N}‐level systems.
\newblock {\em J. Math. Phys.}, 17(5):821--825 (1976).

\bibitem{hall2010}
Michael J.~W. Hall, James~D. Cresser, Li~Li, and Erika Andersson.
\newblock Canonical form of master equations and characterization of
  non-Markovianity.
\newblock {\em Phys. Rev. A}, 89:042120 (2014).

\bibitem{rivas_entanglement_2010}
\'{A}ngel Rivas, Susana~F. Huelga, and Martin~B. Plenio.
\newblock Entanglement and {non}-{Markovianity} of {quantum} {evolutions}.
\newblock {\em Phys. Rev. Lett.}, 105(5):050403 (2010).

\bibitem{jagadish_convex_2020}
Vinayak Jagadish, R.~Srikanth, and Francesco Petruccione.
\newblock Convex combinations of {Pauli} semigroups: {Geometry}, measure, and
  an application.
\newblock {\em Phys. Rev. A}, 101(6):062304 (2020).

\bibitem{jagadish_nonqds_2020}
Vinayak Jagadish, R.~Srikanth, and Francesco Petruccione.
\newblock Convex combinations of {CP}-divisible {Pauli} channels that are not
  semigroups.
\newblock {\em Phys. Lett. A}, 384 (35):126907 (2020).

\bibitem{siudzinskajpa2020}
Katarzyna Siudzi{\'{n}}ska and Dariusz Chru{\'{s}}ci{\'{n}}ski.
\newblock Quantum evolution with a large number of negative decoherence rates.
\newblock {\em J. Phys. A: Math.Theor.}, 53(37):375305 (2020).

\bibitem{hou_singularity_2012}
S.~C. Hou, X.~X. Yi, S.~X. Yu, and C.~H. Oh.
\newblock Singularity of dynamical maps.
\newblock {\em Phys. Rev. A}, 86(1):012101 (2012).

\bibitem{jagadish_measure2_2019}
Vinayak Jagadish, R.~Srikanth, and Francesco Petruccione.
\newblock Measure of not-completely-positive qubit maps: {The} general case.
\newblock {\em Phys. Rev. A}, 100(1):012336 (2019).

\bibitem{jagadish_initial_2021}
Vinayak Jagadish, R.~Srikanth, and Francesco Petruccione.
\newblock Initial entanglement, entangling unitaries, and completely positive
  maps.
\newblock {\em arXiv:2012.12292 [quant-ph]} (2020).

\bibitem{utagi_singularities_2021}
Shrikant Utagi, Vinod~N. Rao, R.~Srikanth, and Subhashish Banerjee.
\newblock Singularities, mixing, and non-{Markovianity} of {Pauli} dynamical
  maps.
\newblock {\em Phys. Rev. A}, 103(4):042610 (2021).

\bibitem{siudzinska_markovian_2021}
Katarzyna Siudzi{\'{n}}ska.
\newblock Markovian semigroup from mixing noninvertible dynamical maps.
\newblock {\em Phys. Rev. A}, 103(2):022605 (2021).

\bibitem{wudarski2016}
Filip~A. Wudarski and Dariusz Chru\ifmmode \acute{s}\else
  \'{s}\fi{}ci\ifmmode~\acute{n}\else \'{n}\fi{}ski.
\newblock Markovian semigroup from non-Markovian evolutions.
\newblock {\em Phys. Rev. A}, 93:042120 (2016).

\bibitem{ruskai2007}
Michael Nathanson and Mary~Beth Ruskai.
\newblock Pauli diagonal channels constant on axes.
\newblock {\em J. Phys. A: Math.Theor.}, 40(28):8171--8204 (2007).


\end{thebibliography}
\end{document}